\documentclass[runningheads,11pt]{llncs}
\usepackage{hyperref}

\newcommand\blfootnote[1]{%
  \begin{NoHyper}%
  \renewcommand\thefootnote{}\footnote{#1}%
  \addtocounter{footnote}{-1}%
  \end{NoHyper}%
}

\usepackage{natbib,amsmath,graphicx,amssymb,xcolor}

\usepackage{amsmath}
\usepackage{float}
\floatstyle{boxed}
\restylefloat{figure}
\usepackage{booktabs} 
\usepackage{mathtools}
\usepackage{bbm} 
\usepackage{geometry} 

\usepackage{rotating}
\newcommand{\amax}{A_{\underu}}
\newcommand{\astar}{A^*}

\newtheorem{cor}{Corollary}[section]

\makeatletter
\DeclareRobustCommand{\qed}{%
  \ifmmode 
  \else \leavevmode\unskip\penalty9999 \hbox{}\nobreak\hfill
  \fi
  \quad\hbox{\qedsymbol}}
\newcommand{\openbox}{\leavevmode
  \hbox to.77778em{%
  \hfil\vrule
  \vbox to.675em{\hrule width.6em\vfil\hrule}%
  \vrule\hfil}}
\newcommand{\qedsymbol}{\openbox}
\newenvironment{proof}[1][\proofname]{\par
  \normalfont
  \topsep6\p@\@plus6\p@ \trivlist
  \item[\hskip\labelsep\itshape
    #1.]\ignorespaces
}{%
  \qed\endtrivlist
}
\makeatother

\newcommand{\abs}[1]{\lvert #1 \rvert}
\def\O{{O}}

\newcommand\undermat[2]{%
  \makebox[0pt][l]{$\smash{\underbrace{\phantom{%
    \begin{matrix}#2\end{matrix}}}_{\text{$#1$}}}$}#2}
\def\mt{{\tilde{m}}}
\def\nt{{\tilde{n}}}
\def\kt{{\tilde{k}}}
\def\Vt{{\mathcal{V}}}
\def\Et{{\mathcal{E}}}
\def\Gt{{\mathcal{G}}}
\def\St{{\mathcal{S}}}

\def\pmin{p_{\text{min}}}
\def\vs{{\mathbf v}}
\newcommand{\opt}{\textsc{OPT}}

\newcommand{\prob}{\text{Pr}}

\begin{document}
\title{Simple Delegated Choice}
\author{Ali Khodabakhsh \inst{1}
\and Emmanouil Pountourakis\inst{2}
\and Samuel Taggart \inst{3}}

\institute{University of Texas at Austin   \\
\email{ali.kh@utexas.edu} \\
\and Drexel University \\\email{manolis@u.drexel.edu} \\
\and Oberlin College   \\\email{staggart@oberlin.edu}
}
\date{}
\maketitle
\begin{abstract}
This paper studies delegation in a model of discrete choice. In the delegation problem, an uninformed principal must consult an informed agent to make a decision. Both the agent and principal have preferences over the decided-upon action which vary based on the state of the world, and which may not be aligned. The principal may commit to a mechanism, which maps reports of the agent to actions. When this mechanism is deterministic, it can take the form of a menu of actions, from which the agent simply chooses upon observing the state. In this case, the principal is said to have delegated the choice of action to the agent.\\

We consider a setting where the decision being delegated is a choice of a utility-maximizing action from a set of several options. We assume the shared portion of the agent's and principal's utilities is drawn from a distribution known to the principal, and that utility misalignment takes the form of a known bias for or against each action. We provide tight approximation analyses for simple threshold policies under three increasingly general sets of assumptions. With independently-distributed utilities, we prove a $3$-approximation. When the agent has an outside option the principal cannot rule out, the constant-approximation fails, but we prove a $\log \rho/\log\log \rho$-approximation, where $\rho$ is the ratio of the maximum value to the optimal utility. We also give a weaker but tight bound that holds for correlated values, and complement our upper bounds with hardness results. One special case of our model is utility-based assortment optimization, for which our results are new.
\blfootnote{Emmanouil Pountourakis was partially supported by  NSF CCF 2218813. Samuel Taggart was partially  supported by NSF CCF 2218814.}
\end{abstract}

\section{Introduction}
\label{sec:intro}
This paper considers a model of delegated stochastic probing.
A decisionmaker (the {\em principal}) must pick one of $n$ actions, each of which yields randomly distributed reward.
Rather than observe rewards directly, the decisionmaker chooses a subset of the actions to allow an agent to consider.
The agent observes the realized rewards exactly, but may be biased towards certain actions and away from others.
The principal's goal is to select a set of actions that will maximize their expected reward from the agent's biased choice.
This template captures a range of economic and managerial dilemmas.
As examples, a firm might seek to replace a piece of expensive equipment, or a national health service must choose which treatment to provide to a patient who might display a range of symptoms. 
The equipment operators know their needs better than managers, and the health service relies on doctors to observe patients. 
In such arrangements, the agent and principal tend not to have preferences which are perfectly aligned: the firm must pay for new equipment (while the operator does not), and specialist doctors might peddle lucrative optional procedures. 

The algorithmic problem above can be couched as mechanism design.
In a revelation mechanism, the agent would observe the actions' rewards and report these to the mechanism, which would choose a possibly randomized action.
The \emph{taxation principle} states that every deterministic mechanism is equivalent to a menu: the principal selects the set of allowable actions, and the agent simply chooses their preferred action upon observing the rewards. Such mechanisms eliminate the need for communication between the agent and principal, and are therefore so common in practice that they are often taken for granted as a managerial tool. 
In economics, {\em delegation} refers exactly to this problem of menu design for a better-informed agent,  coined by \citet{holmstrom1978incentives}.

In the examples above, the alignment of the agent and principals' preferences is well-structured. 
The principal's main uncertainty in the choice problem is payoff-relevant for both parties: in replacing equipment, both the operator and firm want to purchase the right tool for the job. 
Meanwhile, misalignment of preferences is predictable – the firm will pay for the new purchase, and prices are likely known in advance. 
Under these conditions, a particularly salient family of mechanisms is {\em threshold} mechanisms, which restrict the agent to actions where the misalignment of preferences is not too great. 
For our firm and operator, this would take the form of a budget.

\paragraph{Our Contributions.} This work gives a model for delegated choice scenarios like those discussed above. In our model, the agent and principals' preferences for a particular action are captured by two quantities. First, each action $i$ has a shared {\em value} $v_i$, which is unknown to the principal (but distributed according to a known prior) but observable to the agent. Second, each action has a commonly-known and fixed {\em bias} $b_i$, which captures the amount the agents' utility differs from that of the principal. The agent may also have outside options which the principal cannot prohibit; we extend our model to capture this issue as well.

We study three increasingly general regimes, distinguished by the correlation or independence of the value distributions and the absence or presence of an outside option.
For each, we give computational hardness, then take a simple-versus-optimal perspective by completely characterizing the performance of threshold mechanisms. In more detail:
\begin{itemize}
	\item With independently distributed values and no outside option, we show that threshold mechanisms are a $3$-approximation to the optimal mechanism.\footnote{Our results also hold with an outside option if that action has a fixed value, which we make precise subsequently.} We show that this problem is NP-hard.
	\item With independently distributed values and an outside option, threshold mechanisms cannot obtain any nontrivial approximation in general. However, we show a parametrized $\log \rho/\log\log\rho$-approximation, where $\rho$ is the ratio of largest possible value to $\opt$. This problem generalizes the previous problem, and is thus also NP-hard.
	\item With correlation, we give a $\log \pmin^{-1}$ approximation, where $\pmin$ is the probability of the least likely value profile. We show this problem is NP-hard to approximate below a constant factor.
\end{itemize}
We match all three approximation analyses of thresholds with tight examples. A special case of our model is utility-based assortment optimization, a canonical model from revenue management (discussed in Section~\ref{sec:related}). All our results are new to that literature.

\paragraph{Roadmap}

In Section~\ref{sec:prelim}, we give our formal model. We then survey existing work on delegation in Section~\ref{sec:related}, and make specific comparisons to existing work on delegated search and assortment optimization. Section~\ref{sec:indep} contains our hardness result and constant-approximation under independence and lays the groundwork for our parametrized analysis with an outside option in Section~\ref{sec:outside}. Finally, we analyze delegation with correlated values in Section~\ref{sec:res_distortion}.

\section{Model}
\label{sec:model}
\label{sec:prelim}

We now give our model of delegated choice. The principal seeks to choose from a discrete set $\Omega$ of $n$ actions. The principal's utility for action $i$ is given by a random {\em value} $v_i\geq 0$, which the principal is unable to observe. To assist in selecting an action, the principal may consult an agent, who observes all actions' values, and may communicate with the principal after observation. We decompose the agent's utility for action $i$ into its value, shared with the principal, and an unshared {\em bias} term. That is, the agent's utility is given by $u_i=v_i+b_i$. Throughout the paper, we assume each bias $b_i$ is constant and known to the principal.

We assume the principal has the power to commit ex ante to a mechanism for communicating with the agent and selecting an action, and study deterministic mechanisms. By the taxation principle, it suffices to consider mechanisms described by menus over actions. The agent observes all actions' values and selects their utility-maximizing action from the menu --- which may differ from the principal's preferred action. Taking this perspective, we consider the algorithmic problem of selecting a menu $A$ to maximize the principal's expected utility when the agent selects their preferred action according to the observed values. We further assume the existence of an outside option for the agent, denoted action $0$, with value $v_0$ and bias $b_0$. We assume that regardless of the principal's choice of $A$, the agent may always select this action.


Formally, when presented with action set $A\subseteq \Omega$ and after observing the vector of values $\vs$, denote the agent's preferred choice by $g(A,\vs)$. That is, $g(A,\vs)= \text{argmax}_{a\in A\cup \{0\}} (v_i+b_i)$. The principal is faced with a set function optimization problem. We assume the principal has a prior distribution $F$ over the values $\vs$, and must select a menu $A$ for the agent which maximizes their own expected utility.\footnote{We assume the agent breaks ties in the principal's favor, then lexicographically.} That is, the principal solves:
\begin{equation*}
	\underset{A\subseteq \Omega}{\text{maximize}}\,\,\,f(A)\coloneqq 
	\int_{\vs} v_{g(A,\vs)}\,dF(\vs).
\end{equation*}

The model above captures applications such those described in Section~\ref{sec:intro}.
Note that we allow the agent's utility to be negative, and that the model is invariant to additive shifts in the agent's bias for every action.
 We will study a particularly simple set of mechanisms, namely {\em threshold mechanisms}. The threshold mechanism with bias $t$ is given by $A_t=\{i\,|\,b_i\leq t\}$. Note that since the number of threshold policies is at most the number of actions, the principal may compute an optimal threshold efficiently. We analyze the approximation ratio between the best threshold menu and the optimal menu overall.

\begin{example}
	The equipment purchase example described in the introduction can be formulated as follows.
    The firm (principal) needs to buy a piece of equipment, which will be used by a specialist (agent) with knowledge of the quality of different brands.
	Each brand $i$ has quality $q_i$, and price $p_i$.
	Qualities are unknown to the principal, and prices are known.
	We may write values as $v_i=q_i-p_i$ and biases $b_i=p_i$. 
	Note that the values are random, while biases are known, as required.
	A threshold policy restricts to actions with bias ---and hence price--- at most $t$.
\end{example}

\begin{example}
    The health services example from the introduction may be heavily stylized as follows.
    A national health service (principal) needs to select a treatment for a patient with the help of a doctor (agent) who has expertise and observes the patient's condition.
    Each potential treatment $i$ has cost $c_i$ (known to the doctor and the health service), and given the patient's condition an efficacy $e_i$ (known to the doctor but not the health service).
    The health service seeks to maximize the patient's health less costs, $u^P_i=e_i-c_i$.
    The doctor is paid a portion of the costs, and shares some concern for the patient's health.
    For some $\alpha,\beta>0$, we may therefore write $u^A_i=\alpha e_i+\beta c_i$.
    To cast this in our model, note that scaling agent utilities by $1/\alpha$ will not change their decision, so we may normalize $\alpha=1$.
    After normalization, we have $v_i=e_i-c_i$ and $b_i=(\beta+1)c_i$.
    As required, the value $v_i$ depends on $e_i$ and is hence unknown to the principal, and the bias $b_i$ depends only on $c_i$, and is hence known to the principal.
    Further note that a threshold set corresponds to a price cap, restricting the doctor away from the highest-cost procedures.
\end{example}

\section{Related Work}
\label{sec:related}


\paragraph{Simple Versus Optimal Mechanisms.} A primary contribution of computer science to the study of mechanism design is the use of approximation to explain the prevalence of simple mechanisms. For example, \citet{hartline2009simple} prove that the simple auctions often observed in practice can obtain a constant factor of the sometimes-complicated, rarely-used optimal mechanisms. \citet{hartline2013mechanism} surveys similar results for auctions. Recently, \citet{dutting2019simple} and \citet{castiglioni2021bayesian} make similar forays into contract theory, characterizing the power of simple linear contracts. Our work initiates the study of delegated choice through a similar lens.\\

\paragraph{Real-Valued Delegation.} Delegated decisionmaking is a canonical problem in microeconomic theory and managerial science. Much of the literature subsequent to \citet{holmstrom1978incentives} has focused on the special case where the state and action space are continuous and real-valued, and where the preferences of both the agent and principal are single-peaked, but differ by a known bias.
Notable examples include \citet{melumad1991communication}, \citet{martimort2006continuity}, \citet{alonso2008optimal}, and \citet{amador2010optimality}, who characterize the structure of optimal mechanisms under increasingly general variants of the single-peaked model. The main conclusions from these papers are necessary and sufficient conditions for the optimal delegation set to be an interval on the real line. 
Our work makes a similar known bias assumption, but in a model more amenable to algorithmic analysis.
We obtain similar conclusions: the principal can secure high utility by restricting the agent away from extreme actions.

Additional work on similar models includes \citet{kovavc2009stochastic}, who study the gap in performance between randomized and deterministic mechanisms, and \citet{ambrus2017delegation}, who study a principal who can add additional nonmonetary costs to incentivize more preferred decisions.  \citet{aghion1997formal} and \citet{szalay2005economics} consider models in which one or more of the principal and the agent may expend effort to observe a signal about the state. 
For multiple decisions, \citet{frankel2014aligned} considers maxmin robust delegation and \citet{kleiner23} studies Bayesian optimal mechanisms.
 With the exception of \citet{armstrong2010model} and followup works, though, the economics literature has focused on the real-valued model for decisions. Our work considers the mathematically incomparable but similarly common setting of discrete choice. In the latter setting, the structure of the problem renders exact characterization of optimal mechanisms difficult, and motivates the use of a simple-versus-optimal approach.\\



\paragraph{Delegated Search.} The model of delegated project choice from \citet{armstrong2010model} is perhaps closest to ours. The authors consider an agent who chooses between $n$ discrete actions. The principal is able to verify the utilities provided by the selected action, and restrict the agent's choice based on this information. Subsequent followups by \citet{kleinberg2018delegated}, \citet{bechtel2021delegated}, and \citet{bechtel2022delegated} note a strong connection between the \citet{armstrong2010model} model and well-studied online stochastic optimization problems. They upperbound the {\em delegation gap}: they show that even when the agent must pay a search cost to discover each action's utility, the principal can obtain utility within a constant factor of the first-best solution, where they solve the search problem themselves. More recently, \citet{braun2022delegated} give a version where the agent searches online, and make similar comparisons to first-best, and \citet{hajiaghayi23} study a multi-agent version of the model.

Our model differs from the delegated search literature in two notable ways. First is the absence of search. Our agent can perfectly observe the values of all actions. More significantly, our principal is unable to verify the utilities provided by the agent's selected action; they may only rule actions in or out completely. In our model, the first-best solution is $\mathbb E[\max_i v_i]$. The following example shows that no delegation set may approximate the first-best to a factor better than $n$. This contrasts with the constant-approximation results from the work cited above.

\begin{example}
	Consider an instance with $n$ independently-distributed actions. Action $i$ has a value $v_i$ which is $1-\epsilon$ with probability $1/n$ and $0$ otherwise. Each action $i$ has bias $b_i=i$. The first-best expected utility is constant, while in any delegation set, the agent will always pick the highest-indexed action, yielding expected utility $(1-\epsilon)/n$.\\
\end{example}

\paragraph{Stochastic Probing.} There is a by now extensive literature on stochastic probing beyond the economically-inspired settings of this paper and those discussed above. Rather than survey the literature, we offer a few key recent papers, and refer the reader to these for deeper references: \cite{chen2016combinatorial,goel2006asking,mehta2020hitting,segev2021efficient}
Despite similarity of motivation, we employ techniques that largely differ from this literature.\\

\paragraph{Assortment Optimization.} Our model captures special cases of the well-studied {\em assortment optimization} problem. In assortment optimization, a seller must decide which among a set of fixed-price items to offer. A variety of models are common for the buyer's purchase choice, including nested logit models \citep{davis2014assortment,li2015d} and Markov chain-based choice \citep{feldman2017revenue}, along with equivalent models based on random buyer utility \citep{berbeglia2016discrete,AFLS18,aouad2023exponomial}, which includes the especially prevalent multinomial logit (MNL) model as a special case. Our model subsumes assortment optimization with utility-based choice. To see this, consider $n$ items, where the buyer utility $w_i$ for each item $i$ is random, and the revenue $r_i$ for item $i$ is known to the seller. Taking $v_i=r_i+\epsilon w_i$ and $b_i= -(1+\epsilon)r_i$ for sufficiently small $\epsilon>0$ yields an equivalent delegation problem. Under this transformation, an outside option with $v_0=0$ corresponds to the no-buy option, and the option to buy elsewhere with positive utility can be captured with a randomized outside option.

Threshold mechanisms in our model correspond to {\em revenue-ordered assortments}, a well-studied class of solutions for assortment optimization.
A series of papers analyze the approximation ratio of revenue-ordered assortments under increasingly general models: 
\citet{talluri2004revenue} show that revenue-ordered assortments are optimal for several choice models including MNL; \citet{rusmevichientong2014assortment} analyze revenue-ordered assortments for mixtures of MNL models, and further prove NP-hardness of computing the optimal assortment; \citet{berbeglia2020assortment} give parametrized analyses under a general choice model.
Our approximation analyses for independently-distributed values with an random outside option (Sections~\ref{sec:outside}) apply to utility-based assortment optimization, and are new to this literature.
Our logarithmic approximation for correlated values (Section~\ref{sec:res_distortion}) resembles that of \citet{berbeglia2020assortment}; it is less finely parametrized, but extends to more general forms of delegation.

Other work on computational hardness or approximation in assortment optimization includes \citet{desir2020constrained}, who hardness of approximation under a knapsack-constrained version of the problem, and \citet{immorlica2018combinatorial}, who study a version where the buyer has combinatorial preferences over bundles.

\section{Threshold Delegation with Independent Values}
\label{sec:indep}
\newcommand{\imax}{i_{\text{max}}}
\newcommand{\bmax}{b_{\text{max}}}
\newcommand{\underu}{\underline u}

We now consider the simplest case of the model, where the principal's prior $F$ over values is a product distribution, and hence, actions' values are independent. We further assume that the outside option's value, $v_0$, is deterministic, which subsumes the no-outside-option case, as we could have $b_0=-\infty$. We present our approximation result first, and defer hardness to Section~\ref{sec:indephard}

\begin{theorem}\label{thm:indep}
Under independent values and a deterministic outside option, there always exists a threshold mechanism with expected utility that is a $3$-approximation to the optimal deterministic mechanism.
\end{theorem}

\newcommand{\sur}{\textsc{Sur}}
\newcommand{\bdif}{\textsc{BDif}}

 Theorem~\ref{thm:indep} holds regardless of choice of the outside option's fixed value and bias. Before giving the details of the proof, we derive two technical results which will facilitate analysis. In Section~\ref{sec:decomp} for any delegation set, we give a decomposition of the principal's utility into two quantities, one aligned with the agent's utility and one not. Then, in Section~\ref{sec:derand} we use independence obtain a lower bound on the value from threshold sets which will prove useful for both this and the next section's analyses.

\subsection{Utility Decomposition}
\label{sec:decomp}
The principal's task is to balance two sources of utility. On the one hand, when some action has very high value, preferences are aligned: the principal benefits from giving the agent the flexibility to select this action. On the other hand, when actions have smaller values, the principal must control misalignment: they may benefit from restricting the agent away from actions with higher bias, inducing the agent to take actions that provide better value. We now decompose the principal's utility for the optimal delegation set into two quantities, $\sur$ and $\bdif$, which roughly correspond to the value from each of these two cases.

To make the decomposition precise, note that for the optimal delegation set $A^*$, there are two lower bounds imposed on the agent utility from any selection: first, the chosen action must be preferred to the outside option, action $0$, which gives utility at least $b_0$. Second, the agent's utility is at least the bias of the most-biased action in $A^*$. Denote the better of these bounds by $\underu$. We can therefore think of the contribution of any action $i\in A^*$ as decomposing into a {\em bias difference} $\underu - b_i$ and a {\em surplus} $v_i-(\underu-b_i)$. Intuitively, the surplus captures the principal's utility from giving the agent latitude to pick high-valued actions, and the bias difference captures the misaligned portion of the principal's utility. Formally, the decomposition is the following.

\begin{lemma}\label{lem:decom}
	Let $A^*$ denote the optimal delegation set, and let $\underu=\max\{b_i\,|\,i\in A^*\cup\{0\}\}$. 
 Define $\sur$ and $\bdif$ as follows:
	\begin{align*}
		\sur&=\int_{\vs} v_{g(A^*,\vs)}-(\underu-b_{g(A^*,\vs)})\,dF(\vs)\\
		\bdif&=\int_{\vs} \underu-b_{g(A^*,\vs)}\,dF(\vs).
	\end{align*}
	Then we can write $f(A^*)=\sur+\bdif$.
\end{lemma}

To verify the intuition that $\sur$ captures the aligned portion of the principal's utility, note that choosing the smallest threshold set containing all of $A^*\cup\{0\}$ secures $\sur$ for the principal. Formally:

\begin{lemma}\label{lem:sur}
	Let $A_{\underu}=\{i\,|\,b_i\leq \underu\}$. Then $f(A_{\underu})\geq \sur$.
\end{lemma}

\begin{proof}
	We will argue pointwise for each value profile $\vs$. The action chosen by the agent under $\amax$ is $g(\amax,\vs)$, which has $b_{g(\amax,\vs)}\leq \underu$. Since $g(\amax,\vs)$ is the agent's favorite, we have $v_{g(\amax,\vs)}+b_{g(\amax,\vs)}\geq v_{g(\astar,\vs)}+b_{g(\astar,\vs)}$. Hence, \begin{align*}
		v_{g(\amax,\vs)}&\geq v_{g(\astar,\vs)}+b_{g(\astar,\vs)}-b_{g(\amax,\vs)}\\
		&\geq v_{g(A^*,\vs)}-(\underu-b_{g(A^*,\vs)}).
	\end{align*}
	Taking expectation over $\vs$ yields the lemma.
\end{proof}

Lemma~\ref{lem:sur} implies that the main difficulty for obtaining approximately-optimal delegation sets is managing the misaligned portion of the principal's utility. Section~\ref{sec:fixindeppf} gives this analysis for the case with $v_0$ fixed, yielding a $3$-approximation. Note Lemmas~\ref{lem:decom} and \ref{lem:sur} hold even when the outside option's value $v_0$ is randomized. We will therefore make further use of them in our analysis of that case in Section~\ref{sec:outside}.

 \subsection{Lower Bounds via Partial Derandomization}
 \label{sec:derand}
 
To compare the performance of a threshold set $A_t$ to the optimal set $A^*$, we will show that threshold sets can retain sufficient value from $A_t\cap A^*$ without introducing actions in $A_t\setminus A^*$ which overly distort the agent's choices. Independence allows us to summarize the interference of $A_t\setminus A^*$ with a single deterministic action. This will greatly simplify subsequent analyses. 
This section focuses on the case of fixed outside options, but we state our lemma for possibly randomized outside options. We will reuse the result in Section~\ref{sec:outside}.

\begin{lemma}\label{lem:lb}
Assume values are independently distributed. Then for any threshold set $A_t$, there exists a single action $a(t)$ with bias $b_{a(t)}=t$ and 
deterministic value $v_{a(t)}$ such that $f(A_t)\geq f(A_t\cap A^*\cup\{a(t)\})$.
\end{lemma}

Note that $a(t)$ need not be an action from the original delegation instance. The proof will follow from picking the worst realization of actions in $A_t\setminus A^*$ for the principal. Note further that $a(t)$ may differ for every threshold $t$: hence our lower bounds correspond not to one derandomized delegation instance, but to one per threshold.
 \begin{proof}
 For brevity, denote $A_t\setminus A^*$ by $B_t$. Actions in $B_t$ may have randomized values.
 The principal's expected utility $f(A_t)$ can be computed by first realizing $v_i$ for all $i\in B_t$, then computing the principal's expected utility over the values of actions in $G_t=A_t\cap A^*\cup \{0\}$. Hence, there must exist a joint realization of values $\hat v_i$ for each $i\in B_t$ for which this latter expectation is at most $f(A_t)$. Let $\hat B_t$ denote a new set of actions consisting of the actions $i\in B_t$ with $v_i$ fixed as $\hat v_i$. We have $f(G_t\cup B_t)\geq f(G_t\cup \hat B_t)$.  Any actions which are not selected in any realization of the values for $G_t$ may be removed from $\hat B_t$ without consequence. However, since values are fixed for each $i\in \hat B_t$, the agent consistently prefers some particular action $\hat i\in \hat B_t$ over the others in $\hat B_t$. Hence, we may remove all actions but $\hat i$ from $\hat B_t$ without changing the principal's utility.

We finally use this remaining action $\hat i$ to construct $a(t)$. Let $v_{\hat i}$ and $b_{\hat i}$ denote the value and bias of $\hat i$. Define $a(t)$ to have  bias $t$ and value $v_{\hat i}-(t-b_{\hat i})$. Note that $v_{\hat i}+b_{\hat i}=v_{a(t)}+b_{a(t)}$. Hence, the agent will choose $a(t)$ from $G_t\cup \{a(t)\}$ if and only if they would choose $\hat i$ from $G_t\cup\{\hat i\}$. Moreover, $v_{a(t)}= v_{\hat i}$. Hence, $f(G_t\cup\{a(t)\})\leq f(G_t\cup\{\hat i\})\leq f(A_t)$.
 \end{proof}

\subsection{Proof of Theorem~\ref{thm:indep}}
\label{sec:fixindeppf}

We now show how to obtain a $3$-approximation to the optimal delegation utility using threshold mechanisms, assuming $v_0$ is fixed. Lemma~\ref{lem:decom} decomposes the optimal utility into an aligned portion, $\sur$, and a misaligned portion, $\bdif$. Furthermore, Lemma~\ref{lem:sur} states that $\sur$ can be $1$-approximated using a threshold set. Hence, it will suffice to obtain a $2$-approximation to $\bdif$ using thresholds. To do so, we use the derandomization of Lemma~\ref{lem:lb} to derive an even stronger lower bound which holds when $v_0$ is fixed. We then select a threshold for which this lower bound is guaranteed to be large.

\begin{lemma}\label{lem:lowerb}
For any threshold set $A_t$:
\begin{equation*}
	f(A_t)\geq \min\Big(\underu-t,\int_{\vs} (\underu-b_{g(A^*,\vs)})\mathbb I[g(A^*,\vs)\in A_t\cup\{0\}]\,dF(\vs)\Big).
\end{equation*}	
\end{lemma}
To understand our lower bound, note two pitfalls a threshold set could face. First, a too-expansive threshold could include high-bias actions which attract the agent while providing little value. Second, a too-restrictive threshold could leave the agent with too few options. Lemma~\ref{lem:lowerb} states that these are the only two problems: if a threshold $t$ is sufficiently low and includes enough of the actions providing $\bdif$ for $A^*$, $t$ will perform well.


\begin{proof}[Proof of Lemma~\ref{lem:lowerb}]
We argue with respect to the derandomized action $a(t)$. For brevity, write $\underline A_t= 	A_t\setminus A^*\cup\{a(t)\}$. Depending on $v_{a(t)}+b_{a(t)}$, we have two cases, each of which produces a lower bound on $f(\underline A_t)$,
\begin{itemize}
	\item {\em Case 1: $v_{a(t)}+b_{a(t)}< \underu$.} In this case, any time $g(A^*,\vs)\in A_t\cup\{0\}$, we have $g(A^*,\vs)=g(\underline A_t,\vs)$. Since every choice from $A^*$ gives the agent utility at least $\underu$, we have $v_{g(\underline A_t,\vs)}+b_{g(\underline A_t,\vs)}\geq \underu,$ and hence $v_{g(\underline A_t,\vs)}\geq \underu-b_{g(\underline A_t,\vs)}$. Integrating over all $\vs$ yields 
	\begin{equation*}
		f(A_t)\geq f(\underline A_t)\geq \int_{\vs} (\underu-b_{g(A^*,\vs)})\mathbb I[g(A^*,\vs)\in A_t\cup\{0\}]\,dF(\vs).
	\end{equation*}	
	\item {\em Case 2: $v_{a(t)}+b_{a(t)}\geq \underu$.} Then regardless of $\vs$, we have $v_{g(\underline A_t,\vs)}+b_{g(
	\underline A_t,\vs)}\geq \underu$. Since the agent breaks ties in the principal's favor, we also have that $g(\underline A_t,\vs)\neq 0$, so $b_{g(\underline A_t,\vs)}\leq t$. We may conclude that for all $\vs$,  $v_{g(A_t,\vs)}\geq \underu-b_{g(A_t,\vs)}\geq \underu-t$, and hence
\begin{equation*}
	f(A_t)\geq f(\underline A_t)\geq  \underu-t.
\end{equation*}	
\end{itemize}
Hence the lemma holds in both cases.
\end{proof}

The lower bound in Lemma~\ref{lem:lowerb} is a minimum of two terms. We will now study the threshold $\hat t=\underu-\bdif/2$, and observe that both terms in the minimum are at least $\bdif/2$. In particular, we can lower bound the second term as follows:

\begin{lemma}\label{lem:specialt}
    Let $\hat t=\underu-\bdif/2$. Then we have:
    \begin{equation*}
        \int_{\vs} (\underu-b_{g(A^*,\vs)})\mathbb I[g(A^*,\vs)\in A_{\hat t}\cup\{0\}]\,dF(\vs)\geq \bdif/2.
    \end{equation*}
\end{lemma}
\begin{proof}
Let $\mathcal E_2$ denote the event that $g(A^*,\vs)\in A_{\hat t}\cup\{0\}$. The following sequence of inequalities, explained below, implies the lemma:
\begin{align*}
    \bdif&=\int_{\vs} \underu-b_{g(A^*,\vs)}\,dF(\vs)\\
    &=\int_{\vs} (\underu-b_{g(A^*,\vs)})\mathbb I[\mathcal E_2]\,dF(\vs)+\int_{\vs} (\underu-b_{g(A^*,\vs)})\mathbb I[\overline{\mathcal E_2}]\,dF(\vs)\\
    &<\int_{\vs} (\underu-b_{g(A^*,\vs)})\mathbb I[\mathcal E_2]\,dF(\vs)+(\underu-\hat t)\\
    &=\int_{\vs} (\underu-b_{g(A^*,\vs)})\mathbb I[\mathcal E_2]\,dF(\vs)-\bdif/2.
\end{align*}
The first equality is the definition of $\bdif$. The third equality follows from the fact that under $\overline{\mathcal E_2},$ $b_{g(\astar,\vs)}>\hat t$, and from the fact that this occurs with probability at most $1$. The last equality follows from the definition of $\hat t$.
\end{proof}

\begin{proof}[Proof of Theorem~\ref{thm:indep}]
By combining Lemma~\ref{lem:lowerb}, with the definition of $\hat t$ and Lemma~\ref{lem:specialt}, we have:
\begin{align*}
	f(A_{\hat t})&\geq \min\Big((\underu-\hat t),\int_{\vs} (\underu-b_{g(A^*,\vs)})\mathbb I[g(A^*,\vs)\in A_{\hat t}\cup\{0\}]\,dF(\vs)\Big)\\
	&=\min\Big(\tfrac{\bdif}{2},\int_{\vs} (\underu-b_{g(A^*,\vs)})\mathbb I[g(A^*,\vs)\in A_{\hat t}\cup\{0\}]\,dF(\vs)\Big)\\
	&\geq \min\left(\tfrac{\bdif}{2},\tfrac{\bdif}{2}\right)=\tfrac{\bdif}{2}.
\end{align*}
The theorem now follows from noting that $f(A^*)=\sur+\bdif\leq f(\amax)+2f(A_{\hat t}).$
\end{proof}

The proof of Theorem~\ref{thm:indep} used independence once, in the derandomization step of Lemma~\ref{lem:lb}. Nevertheless, we show in Section~\ref{sec:res_distortion} that independence is critical to guaranteeing the performance of threshold mechanisms by giving a super-constant lower bound in its absence. With independence, the following example matches the upper bound exactly:

\begin{example}
Our example will have five actions, with biases and value distributions given below. The outside option will have $b_0=-\infty$, and therefore can be ignored. Take two small numbers, $\delta$ and $\epsilon$, with $\delta$ much smaller than $\epsilon$. Actions will be as follows:
\begin{itemize}
	\item $b_1=0$. $v_1$ is $1+2\delta$ with probability $\epsilon$, and $0$ otherwise.
	\item $b_2=1-\epsilon-\delta$. $v_2=4\delta+\epsilon$.
	\item $b_3=1-\epsilon$. $v_3=\epsilon+\delta$.
	\item $b_4=1-\delta$. $v_4=5\delta$.
	\item $b_5=1$. $v_5$ is $1$ with probability $\epsilon$, and $0$ otherwise.
\end{itemize}
We may analyze the instance neglecting $\delta$ terms, which only serve to break ties for the agent. The optimal delegation set is $\{1,3,5\}$, with principal utility $(1-(1-\epsilon)^2)+\epsilon(1-\epsilon)^2$, where the first term comes from the event that either actions $1$ or $5$ realize their high values (in which case they are chosen), and the second term comes from the event that $1$ and $5$ are low-valued, in which case the agent prefers action $3$. As $\epsilon\rightarrow 0$, the optimal value goes to $0$ as $\approx3\epsilon$. Meanwhile, no threshold obtains expected value better than $\epsilon$. This yields an approximation ratio of $3$ in the limit.
\end{example}

\subsection{Computational Hardness}
\label{sec:indephard}

We conclude the section by discussing the complexity of the delegation problem with independent values.
For the discrete version of the problem, where every action $i$ is specified by a bias $b_i$ and a list of realizations $((v_i^1,p_i^1),\ldots,(v_i^n,p_i^n))$, we prove:

\begin{theorem}\label{thm:indephard}
    Delegation with independent values is NP-complete, even with no outside option.
\end{theorem}

The challenge in proving NP-hardness is managing the rigid structure of the joint value distribution imposed by independence.
We adopt a similar strategy to \citet{chen2014complexity}, who show that pricing to a unit-demand buyer is hard.
We reduce from \textsc{Integer Partition}: given integers $c_1,\ldots,c_n$, the goal is to find a subset $S\subseteq [n]$ such that $\sum_{i\in S}c_i=\tfrac{1}{2}\sum_{i=1}^n c_i$. 
We associate each integer $c_i$ with an action $i$. 
Each such action impacts the principal's utility via two low-probability realizations: a bad realization which harms the principal's utility and a good realization which improves the principal's utility.
These low probabilities are tuned in such a way that only first- and second-order terms in the probability calculation are relevant.
Furthermore, the tuning is such that the bad events scale linearly with the $c_i$s, while the good events scale in a concave way, with the principal's utility being maximized when actions taken correspond to an even split of the integers. 
Full details can be found in Appendix~\ref{app:indephard}.
Note that Theorem~\ref{thm:indephard} also implies hardness of the model in the next section, with a random outside option.

\section{Randomized Outside Options}
\label{sec:outside}
\newcommand{\vmax}{v_{\max}}

In Section~\ref{sec:indep}, we showed that with a fixed (or non-existent) outside option, a simple delegation set secures a constant fraction of the utility from the optimal delegation set. We now consider the case where the outside option's value is randomized. This may be more realistic in scenarios such as assortment optimization, where the agent's outside option is taking an action (i.e.\ buying a good) somewhere else. In this regime, we again give tight bounds. In Section~\ref{sec:unparam}, we show that no nontrivial multiplicative approximation is possible with threshold sets: there are examples where thresholds give no better than an $\Omega(n)$-approximation, which can be matched trivially. However, in Section~\ref{sec:param} we show that such lower bound examples are necessarily unnatural. In particular, we parametrize our analysis by the ratio $\rho=\vmax/\opt$, where $\opt$ is the optimal principal utility and $\vmax$ the highest value in any action's support. We prove that the worst-case approximation is $\Theta(\log \rho/\log\log \rho)$: hence, whenever thresholds perform poorly, it is because the optimal solution relies on exponentially large, exponentially rare values.

\subsection{Unparametrized Analysis: Impossibility}
\label{sec:unparam}

This section gives an unparametrized analysis of threshold delegation with a randomized outside option. We show that it is not possible to guarantee a nontrivial approximation factor which holds across all instances.

Our constant-approximation in Section~\ref{sec:indep} relied on our ability to separate the optimal utility into two parts, $\bdif$ and $\sur$. Approximating the bias difference $\bdif$ was the crux of the analysis. The following example shows that with a random outside option value $v_0$, this analysis --- and in particular the approximation of $\bdif$ --- fails. We will choose our distribution over $v_0$ to streamline exposition, but the example that follows could be adjusted so that the distribution over $v_0$ satisfies nearly any desired regularity condition.

\begin{example}\label{ex:outside}
Our example will feature two sets of actions: {\em good} actions, which are taken by the optimal delegation set, and {\em bad} actions, which are not. We will index the actions so that the $i$th good action is $g(i)$, and the bad action between good actions $g(i-1)$ and $g(i)$ is $b(i)$. For $i\in\{1,\ldots, n\}$ $g(i)$ will have:
\begin{itemize}
	\item[\textbullet] bias $n^{n-1}-n^{n-i}$:
	\item[\textbullet] value $n^{n-i}+i\epsilon$ with probability $1/n$, and $0$ otherwise.\footnote{It is equivalent for this example to use distribution $n^{n-i}+i\epsilon$ with probability $1/n$, and $n^{n-i}-\epsilon$ otherwise. Under this distribution, the example becomes an instance of assortment optimization, as described in Section~\ref{sec:related}. This makes our parametrized analysis in Section~\ref{sec:param} tight even for that special case.}
\end{itemize}
The bad actions will be indexed by $b(i)$ for $i\in\{2,\ldots, n\}$. Bad action $b(i)$ will have
\begin{itemize}
	\item[\textbullet]bias $n^{n-1}-n^{n-i}$.
	\item[\textbullet] value $n^{n-i}+(i-1)\epsilon+\delta$, for $\delta\ll \epsilon$.
\end{itemize}

The outside option will have bias $n^{n-1}$ and value $v_0$ distributed according to a discrete distribution. We will set $\Pr[v_0=\epsilon/2]=n^{-(n-1)}$. For $i>1$, we will choose probability mass function $\Pr[v_0=i\epsilon-\epsilon/2]=n^{-(n-i)}-n^{-(n-i+1)}$. Note that we have picked these probabilities so that $\Pr[v_0<i\epsilon]=n^{-(n-i)}$. The values of all actions described above are independent.

A solution to the delegation instance we just described is to take only good actions. The probability that at least one good action takes its high value is $1-(1-1/n)^n\geq 1-1/e$. Assume this event has occurred, and that the agent's preferred good action is $g(i)$. Then $g(i)$ is preferred to the outside option with probability $n^{-(n-i)}$. Hence, the principal's expected utility from choosing only good actions is at least:
\begin{equation*}
    f(\{g(1),\ldots,g(n)\})\geq (1-1/e)n^{-(n-i)}(n^{n-i}+\epsilon i)\geq 1-1/e.
\end{equation*}

Now consider a threshold set $A_t$. It is without loss of generality to consider $t=n^{n-1}-n^{n-j}$ for some $j$, which implies that the highest-bias actions in $A_t$ are $g(j)$ and $b(j)$. For any good action $g(i)$, with $i<j$, the agent's utility for $g(i)$ on a high-valued realization is $n^{n-1}+i\epsilon<n^{n-1}+(j-1)\epsilon+\delta$. Hence, the agent ignores all actions other than $g(j)$, $b(j)$, and the outside option. If $g(j)$ draws its high value, the principal gets utility $n^{n-j}+j\epsilon$ utility if and only if $g(j)$ survives the outside option, which happens with probability $n^{-(n-j)}$. Otherwise, the agent looks to action $b(j)$, and takes it over the outside option with probability $n^{-(n-j+1)}$. Ignoring value from the outside option, which goes to $0$ as $\epsilon\rightarrow 0$, we can account for the utility from $A_t$ as follows:
\begin{align*}
    f(A_t)&=\tfrac{1}{n}n^{-(n-j)}(n^{n-j}+j\epsilon)+(1-\tfrac{1}{n})(n^{n-j}+(j-1)\epsilon+\delta)n^{-(n-j+1)}.\\
    &\approx \tfrac{1}{n}+(1-1/n)\tfrac{1}{n},
\end{align*}
where the latter approximation holds for $\epsilon$ and $\delta$ sufficiently small. This implies that every threshold incurs a loss which is $\Omega(n)$.
\end{example}

An upper bound of $n$ for threshold mechanisms is trivial, by the following lemma. Hence, up to a constant, the lower bound in Example~\ref{ex:outside} is tight.

\begin{lemma}\label{lem:single}
	For any set $A$ and $i\in A\cup\{0\}$, let $A^i=\int_{\vs}v_i\,\mathbb I[g(A,\vs)=i]\,d\vs$ denote the contribution to $f(A)$ from action $i$. Then $f(A_{b_i})\geq A^i$.
\end{lemma}
\begin{proof}
Consider any $\vs$ where $g(A,\vs)=i$. The action chosen by the agent under $A_{b_i}$ is $g(A_{b_i},\vs)$, which has $b_{g(A_{b_i},\vs)}\leq b_i$. 
Since $g(A_{b_i},\vs)$ is the agent's favorite, we have $v_{g(A_{b_i},\vs)}+b_{g(A_{b_i},\vs)}\geq v_{g(A,\vs)}+b_{g(A,\vs)}=v_i+b_i$. 
Hence, $v_{g(A_{b_i},\vs)}\geq v_i+b_i-b_{g(A_{b_i},\vs)}
\geq v_i$. Taking expectation over $\vs,$ we obtain:
\begin{equation*}
	f(A_{b_i})=\int_{\vs}v_{g(A_{b_i},\vs)}\,dF(\vs)\geq \int_{\vs}v_{g(A_{b_i},\vs)}\,\mathbb I[g(A,\vs)=i]\,dF(\vs)\geq \int_{\vs}v_i\,\mathbb I[g(A,\vs)=i]\,dF(\vs),
\end{equation*}
where the first inequality follows from the nonnegativity of $v_i$.
\end{proof}

An $n$-approximation then follows from noting that for any set $A$, $f(A)=\sum_i A^i$.
\begin{cor}
	With independent values (and possibly randomized outside option), the best threshold is an $n$-approximation to the optimal delegation set.
\end{cor}

\subsection{Parametrized Approximation}
\label{sec:param}

In the previous section, we gave an example where no threshold set was better than an $\Omega(n)$-approximation. However, this example was extreme, in the sense that while the optimal solution obtained $O(1)$ utility, some actions had values which were as large as $n^{n-1}$. We now show that this is no coincidence: any example where threshold mechanisms perform poorly must be unnatural in this way.

\begin{theorem}\label{thm:param}
	Let $\rho=\vmax/\opt$, where $\opt$ is the optimal principal utility and $\vmax$ the highest value in any action's support. Then with independent values (and a possibly randomized outside option), the best threshold is a $O(\log \rho/\log\log \rho)$-approximation to $\opt$.
\end{theorem}

Theorem~\ref{thm:param} is of particular interest for the application of assortment optimization. For an instance of the latter problem, each item $i$ yields revenue $p_i$ for the seller, and value $w_i$ for the buyer. Framed as a delegation problem, we have $v_i=p_i+\epsilon w_i$, for sufficiently small $\epsilon$. Hence, Theorem~\ref{thm:param} implies that the prices $p_i$ must be extreme whenever revenue-ordered assortments perform poorly. Another consequence of Theorem~\ref{thm:param} is a bicriteria approximation when values lie in $[0,1]$: either some threshold obtains a small multiplicative approximation, or it is trivial to obtain an additive approximation.



\begin{proof}[Proof of Theorem~\ref{thm:param}]
We will argue the contrapositive.
That is, we will argue with respect to some integer $\alpha\geq4$, and assume that no threshold obtains a $\beta$-approximation for any $\beta<16\alpha$.
Under this assumption, we show that there must be an action with value at least $(\alpha-2)^{\alpha-1}\opt/8\alpha$ with positive probability.
The analysis will roughly proceed in three steps. 
First, we partition the optimal solution into $\alpha$ subsets with roughly equal contribution to $\opt$.
We then consider the thresholds based on each of these subsets, and compare their utility to that from the sets themselves; by assumption, no such threshold will outperform its respective subset.
Finally, we combine the resulting inequalities to show that the only way all can hold simultaneously is if the bias of one of these thresholds is extreme. 
This will imply the existence of a comparably high value.
Throughout, we will make use of our decomposition and derandomization from Lemmas~\ref{lem:decom} and \ref{lem:lb}, respectively.\\

\paragraph{Decomposing \opt} 
Before defining our thresholds, we note that when no threshold approximates $\opt$ well, we may draw several simplifying conclusions about the structure of $\opt$.
Let $A^*$ be an optimal subset of actions, and assume every action in $A^*$ is selected with positive probability.
Following Lemma~\ref{lem:decom}, write $\opt=f(A^*)=\sur+\bdif$, and write $\underu=\max\{b_i\,|\,i\in A^*\cup\{0\}\}$.
By Lemma~\ref{lem:sur}, it must be that $\sur< \opt/\alpha$, or else the grand threshold $\amax$ would be an $\alpha$-approximation, contradicting the nonexistence of any $\beta<16\alpha$-approximation. 
We may therefore focus our analysis on $\bdif=\int_{\vs} \underu-b_{g(A^*,\vs)}\,dF(\vs)$.
It must again be that no threshold obtains better than an $8\alpha$-approximation to $\bdif>(1-1/\alpha)\opt\geq\opt/2$.


Next, note that no one action can comprise a large fraction of $\bdif$. More precisely, let
\begin{align*}
	\opt^i=\int_{\vs}v_i\mathbb{I}[g(A^*,\vs)=i]\,dF(\vs)\\
	\bdif^i=\int_{\vs}(\underu-b_i)\mathbb{I}[g(A^*,\vs)=i]\,dF(\vs)
\end{align*}
denote the contribution of action $i$ to $\opt$ and $\bdif$, respectively. Since $g(A^*,\vs)=i$ only if $v_i\geq(\underu-b_i)$, we must have $\opt^i\geq \bdif^i$. 
Lemma~\ref{lem:single} implies that we may obtain $\opt^i$ from a threshold set for any $i\in A^*\cup\{0\}$.
We must therefore have that $\bdif^0\leq\bdif/8\alpha$, and therefore that $\bdif-\bdif^0\geq (1-1/8\alpha)\bdif\geq \bdif/2$.
The remainder of the proof will focus on approximating $\bdif-\bdif^0$, assuming no approximation better than $4\alpha$ is possible.
Lemma~\ref{lem:single} implies that $\bdif^i\leq(\bdif-\bdif^0)/4\alpha$ for all $i\in A^*$.\\

\paragraph{Constructing Thresholds}
We now obtain a sequence of candidate thresholds by partitioning the actions in $A^*$ based on their contribution to $\bdif-\bdif^0$. 
Let $m=|A^*|$, and relabel the actions in $A^*$ as $1,\ldots,m$, with $b_i\leq b_{i+1}$ for all $i\in\{1,\ldots,m-1\}$. 
Actions not in $A^*$ will be indexed $m+1,\ldots,n$ (with $0$ keeping its label).
We will partition $A^*$ greedily to produce $\alpha$ subsets of roughly equal contribution to $\bdif-\bdif^0$. 
More precisely, define the first breakpoint between subsets as $k_0=0$, and for $j\in \{1,\ldots,\alpha\}$, define subsequent breakpoints $k_j$ recursively as the smallest $k\in\{k_{j-1}+1,\ldots,m\}$ such that $\sum_{i=1}^k\bdif^i\geq j(\bdif-\bdif^0)/\alpha$. 
Define the partition as $A^*(j)=\{k_{j-1}+1,\ldots,k_j\}$ for all $j$.
We can upper- and  lowerbound the distortion of each bin:
since $\bdif^i<(\bdif-\bdif^0)/4\alpha$ for all $i$, it must be that $A^*(j)$ is nonempty for $j\in \{1,\ldots,\alpha\}$.
Further define $
\bdif(j)=\sum_{i\in A^*(j)}\bdif^i$. 
We must also have $\bdif(j)\geq (\bdif-\bdif^0)/2\alpha$ for all $j$, and hence no threshold can attain better than a $2$-approximation to $\bdif(j)$ for any $j$.
We will define candidate thresholds based on each bin: let $t_j=b_{k_j}$ (with $t_0=\min_{i\in A^*} b_i$), and define the threshold sets $A_j=\{i\,|\,b_i\leq b_{k_j}\}$ for all $j$.\\

\paragraph{Lower Bounding Threshold Utilities}
\newcommand{\inj}{\mathcal E_j}
\newcommand{\outj}{\overline {\mathcal E}_j}
\newcommand{\agreej}{{\mathcal E}^=_j}
\newcommand{\disj}{\overline {\mathcal E}^=_j}
\newcommand{\insidej}{\mathcal E^+_j}
\newcommand{\survj}{\mathcal E_j^\leq}
\newcommand{\pr}{\text{Pr}}
\newcommand{\domj}{\mathcal E_j^*}
\newcommand{\optsur}{\mathcal E^*}
To lower bound the principal utility from $A_j$, we apply Lemma~\ref{lem:lb} and consider $f(A_j\cap A^*\cup\{a(t_j)\})$ for some suitably constructed interfering action $a(t_j)$ with bias $t_j$.
We write $\underline A_j=A_j\cap A^*\cup\{a(t_j)\}$. 
Note that the interfering action $a(t_j)$ must give the agent high utility, or else $\underline A_j$ would perform as well as its respective bin $A^*(j)$. Formally, let $\inj$ be the event that the agent's preferred action from $A^*$ is in $A^*(j)$.
(Note that the agent may still ultimately choose action $0$ under either $\inj$ or $\outj$, and that $\sum_j\pr[\inj]=1$.)
If $v_{a(t_j)}+b_{a(t_j)}< \underu$, then the threshold performance $f(\underline A_j)$ can be written as
\begin{equation*}
	\int_{\vs}v_{g(\underline A_j,\vs)}\,dF(\vs)
	\geq \int_{\vs}v_{g(\underline A_j,\vs)}\mathbb I[\inj]\,dF(\vs)
	=\int_{\vs}v_{g(A^*,\vs)}\mathbb I[\inj]\,dF(\vs)
	\geq\int_{\vs}(\underu-b_{g(A^*,\vs)})\mathbb I[\inj]\,dF(\vs).
\end{equation*}
The righthand expression is the bias difference conditioned on $\inj$, which is at least $\bdif(j).$ 
This contradicts our assumption that no threshold approximates $\bdif(j)$ to a factor better than $2$.

%
%

We will now lowerbound $f(\underline A_j)$ by decomposing it into two terms, depending on the event $\inj$:
\begin{equation}\label{eq:threshlb}
	f(\underline A_j)=\int_{\vs}v_{g(\underline A_j,\vs)}\,dF(\vs)
	=\int_{\vs}v_{g(\underline A_j,\vs)}\mathbb I[\inj]\,dF(\vs)+\int_{\vs}v_{g(\underline A_j,\vs)}\mathbb I[\outj]\,dF(\vs)
\end{equation}
To lowerbound the first term the right of (\ref{eq:threshlb}), define the event $\insidej$ to be the event that $g(\underline A_j,\mathbf v)\neq 0$, and $\agreej$ to be the event that the agent's favorite action in $A^*(j)$ is also their favorite in $\underline A_j$. In $\agreej$, the agent may still ultimately take the outside option.
We may now write:
\begin{equation*}
	\int_{\vs}v_{g(\underline A_j,\vs)}\mathbb I[\inj]\,dF(\vs)=\int_{\vs}v_{g(\underline A_j,\vs)}\mathbb I[\inj\cap\insidej\cap\agreej]\,dF(\vs)
	\geq \int_{\vs}(\underu-b_{g(A^*,\vs)})\mathbb I[\inj\cap\insidej\cap\agreej]\,dF(\vs),
\end{equation*}
where the inequality comes from the fact that under $\inj\cap\insidej\cap\agreej$, the agent chooses the same thing from $A^*$ as they would from $\underline A_j$, and hence that action gives the agent utility at least $\underu$.

To lowerbound the second term on the right of (\ref{eq:threshlb}), let $\survj$ be the event that the agent prefers $a(t_j)$ to the outside option, i.e.\  $v_0+b_0\leq v_{a(t_j)}+b_{a(t_j)}$. We may then lowerbound the second term as:
\begin{equation*}
	\int_{\vs}v_{g(\underline A_j,\vs)}\mathbb I[\outj\cap\insidej]\,dF(\vs)\geq \int_{\vs}(\underu-t_j)\mathbb I[\outj\cap\insidej]\,dF(\vs)
	\geq(\underu-t_j)\pr[\outj]\pr[\survj].
\end{equation*}
The first inequality follows from the fact that the action chosen in $\underline A_j$ yields utility at least $v_{a(t_j)}+t_j$ and has bias at most $t_j$. The second inequality follows from the facts that $\survj\leq\insidej$ and $\survj$ and $\inj$ are independent. Taking the two terms together, we have the lower bound
\begin{equation}\label{eq:bothlower}
	f(\underline A_j)\geq \int_{\vs}(\underu-b_{g(A^*,\vs)})\mathbb I[\inj\cap\insidej\cap\agreej]\,dF(\vs)+(\underu-t_j)\pr[\outj]\pr[\survj].
\end{equation} 

\begin{figure}[b]
	\caption{List of events for proof of Theorem~\ref{thm:param}.}
	\begin{itemize}
		\item $\inj$: agent's favorite non-$0$ action in $A^*$ is in $A^*(j)$.
		\item $\insidej$: $g(\underline A_j,\vs)\neq 0$.
		\item $\agreej$: the agent's favorite action in $A^*(j)$ is also their favorite in $\underline A_j$.
		\item $\survj$: $v_0+b_0\leq v_{a(t_j)}+b_{a(t_j)}$.
		\item $\optsur$: $g(A^*,\vs)\neq 0$
	\end{itemize}
\end{figure}

\paragraph{Upper Bounding $\bdif(j)$} Next, we upperbound the contribution of $A^*(j)$ to $\bdif$. We will split $\bdif(j)$ based on the event $\agreej$ that the agent's favorite action is the same between $A^*(j)$ and $\underline A_j$. Let $\optsur$ denote the event that $g(A^*,\vs)\neq 0$. We have:
\begin{align}
	\bdif(j)&=\int_{\vs}(\underu-b_{g(A^*,\vs)})\mathbb I[\inj\cap \optsur]\,dF(\vs)\notag\\
	&=\int_{\vs}(\underu-b_{g(A^*,\vs)})\mathbb I[\inj\cap \optsur\cap \agreej]\,dF(\vs)+\int_{\vs}(\underu-b_{g(A^*,\vs)})\mathbb I[\inj\cap\optsur\cap \disj]\,dF(\vs).\label{eq:bdifub}
\end{align}

We can now upperbound each term in (\ref{eq:bdifub}), starting by rewriting the leftmost. In  the intersection event $\inj\cap\agreej$, the favorite non-$0$ action from $A^*$ and $\underline A_j$ are the same. Hence conditioned on this event, $g(\underline A_j,\mathbf v)=g(A^*,\mathbf v)
$ and $\optsur=\insidej$. Therefore:
\begin{equation*}
	\int_{\vs}(\underu-b_{g(A^*,\vs)})\mathbb I[\inj\cap \optsur\cap \agreej]\,dF(\vs)
	=\int_{\vs}(\underu-b_{g(A^*,\vs)})\mathbb I[\inj\cap \insidej\cap \agreej]\,dF(\vs).
\end{equation*}
To upperbound the second term in (\ref{eq:bdifub}), note first that in the event $\inj\cap\optsur$, the chosen action's bias is at least $t_{j-1}$ and hence the bias difference is at most $\underu-t_{j-1}$.
Second, note that conditioned on $\inj\cap \disj$, it must be that the agent prefers $a(t_j)$ to all actions in $A^*(j)$. Hence, conditioned on $\inj\cap \disj$, it must be that any time the agent prefers an action in $A^*(j)$ to the outside option, it must be that they also prefer $a(t_j)$. Hence, $\inj\cap\optsur\cap \disj\subseteq \inj\cap\survj\cap \disj$.
Finally, note  that $\inj$ and $\survj$ are independent.
These facts imply:
\begin{align*}
	\int_{\vs}(\underu-b_{g(A^*,\vs)})\mathbb I[\inj\cap\optsur\cap \disj]\,dF(\vs)
	&\leq(\underu-t_{j-1})\int_{\vs}\mathbb I[\inj\cap\optsur\cap \disj]\,dF(\vs)\\
	&\leq(\underu-t_{j-1})\int_{\vs}\mathbb I[\inj\cap\survj\cap \disj]\,dF(\vs)\\
	&\leq(\underu-t_{j-1})\pr[\inj\cap \survj]\\
	&=(\underu-t_{j-1})\pr[\inj]\pr[ \survj].
\end{align*}
Combining the two terms, we have:
\begin{equation}\label{eq:bothupper}
	\bdif(j)\leq \int_{\vs}(\underu-b_{g(A^*,\vs)})\mathbb I[\inj\cap \insidej\cap \agreej]\,dF(\vs)+(\underu-t_{j-1})\pr[\inj]\pr[ \survj].
\end{equation}

\paragraph{Lowerbounding Bias Differences}
Inequality (\ref{eq:bothlower}) lowerbounds $f(\underline A_j)$ in terms of $t_j$, and inequality (\ref{eq:bothupper}) upperbounds $\bdif(j)$ in terms of $t_{j-1}$. 
Since no threshold is better than a $2$-approximation to $\bdif(j)$, it must hold that our upper bound on $\bdif(j)$ exceeds our lower bound on $f(\underline A_j)$. That is:
\begin{align*}
&	\int_{\vs}(\underu-b_{g(A^*,\vs)})\mathbb I[\inj\cap\insidej\cap\agreej]\,dF(\vs)+(\underu-t_j)\pr[\outj]\pr[\survj]\\
&\quad\quad\quad\leq \int_{\vs}(\underu-b_{g(A^*,\vs)})\mathbb I[\inj\cap \insidej\cap \agreej]\,dF(\vs)+(\underu-t_{j-1})\pr[\inj]\pr[ \survj].
\end{align*}
We may rearrange this as 
\begin{equation*}
	\frac{\underu-t_{j-1}}{\underu-t_j}\geq \frac{1-\pr[\inj]}{\pr[\inj]}.
\end{equation*}
Taking the product over all $j\leq \alpha-1$ and canceling yields:
\begin{equation*}
	\frac{\underu-t_{0}}{\underu-t_{\alpha-1}}\geq \prod_{j=1}^{\alpha-1}\frac{1-\pr[\inj]}{\pr[\inj]}.
\end{equation*}
Note that the righthand side is a convex, symmetric function of the $\pr[\inj]$s. Moreover, we have $\sum_{j=1}^{\alpha-1}\pr[\inj]\leq \sum_{j=1}^{\alpha}\pr[\inj]=1$. Hence, minimizing the righthand side as a function of the $\pr[\inj]$s yields a minimum at $\pr[\inj]=1/(\alpha-1)$ for all $j$, and hence $(\underu-t_{0})/(\underu-t_{\alpha-1})\geq (\alpha-2)^{\alpha-1}$. Note also that $\bdif(\alpha)\geq (\bdif-\bdif^0)/2\alpha$. Since $\bdif(\alpha)\leq (\underu-t_{\alpha-1})\pr[\mathcal E_{\alpha}]\leq(\underu-t_{\alpha-1})$, we obtain:
\begin{equation*}
	\underu-t_{0}\geq \frac{(\alpha-2)^{\alpha-1}}{2\alpha}(\bdif-\bdif^0).
\end{equation*}
Since every action in $A^*$ is selected with positive probability, and since $t_0=\min_{i\in A^*} b_i$, it must be that some action in $A^*$ has value at least  $\underu-t_{0}$ with positive probability. Since $\bdif-\bdif^0\geq \bdif/2\geq \opt/4$, we obtain the desired lower bound on $\vmax/\opt$.

\end{proof}

\section{Threshold Delegation with Correlated Values}
\label{sec:res_distortion}
In the previous sections, we showed that under independently-distributed values, simple threshold rules obtain a close approximation the optimal principal utility. We now allow arbitrarily correlated values and show that the situation worsens considerably. Assuming the value distribution is discrete, prove tight a approximation guarantee for the principal's best threshold policy, showing that it is a $\Theta(\log \pmin^{-1})$-approximation, where $\pmin$ denotes the lowest probability mass of any value profile realization. Hence, absent independence, threshold policies still perform well under low levels of uncertainty, but their performance gradually degrades as the uncertainty grows more extreme. We state our results formally below, starting with our upper bound.

\begin{theorem}\label{thm:uniform}
	
	There always exists a threshold policy which is a $4\log(\pmin^{-1})$-approximation
	where $\pmin$ 
	is the mass of the least likely value profile.
\end{theorem}

\begin{proof}
Let $OPT$ be the optimal delegation set, and let $t_0$ be the maximum bias across actions in $OPT$. Let $B$ be the random variable that corresponds to the bias of the action chosen in $OPT$. Let $t_1$ be the bias threshold  such that $\Pr[B\in [t_1,t_0]] \leq \frac{1}{2}$ and $\Pr[B\in [0,t_1]] \geq \frac{1}{2}$.
Let $OPT([t_1,t_0])$ be the principal utility generated conditioned on $B\in [t_1,t_0]$.
 We claim that the principal utility generated by using a best threshold out $A_{t_1}$ or $A_{t_0}$ achieves principal utility at least $\frac{1}{4} OPT([t_1,t_0])$.
 
 First consider the threshold $A_{t_0}$. Consider any realization where $OPT$ picks an action with bias at least $t_0$. Let $v,b$ be the value and bias of the action chosen by $OPT$ and let $v',b'$ be the value and bias of the action chosen by $A_{t_0}$ respectively.   
 Since the action chosen by $OPT$ is available in $A_{t_0}$ it must be that 
 \begin{equation*}
     v'+b' \geq v+b \Leftrightarrow v-v' \leq b'-b.
 \end{equation*}
 Note that from our assumptions $b'\leq t_0$ and $b\geq t_1$, therefore the pointwise loss of $A_{t_0}$ compared to $OPT$ is at most $t_0-t_1$ in this event.
 As a result we can lower bound the principal utility of $A_{t_0}$ as follows:
 $$
f(A_{t_0}) \geq \Pr[B \in[t_1,t_0]] (OPT([t_1,t_0])  -  t_0+t_1) \geq  \frac{1}{2}(OPT([t_1,t_0]) - t_0+t_1).
 $$
 
 Second, consider the threshold $A_{t_1}$. Every time $OPT$ chooses an action with bias less than or equal to $t_1$ this action is also available to $A_{t_1}$. Note that if such action is chosen its agent utility must be at least $t_0$ otherwise the action with the maximum bias would have been chosen instead. The action chosen in $A_{t_1}$ therefore must have at least agent utility $t_0$. Since the bias is at most $t_1$ this means that the principal utility from that action is at least $t_1-t_0$. We conclude that
$$
f(A_{t_1}) \geq  (t_0-t_1)\Pr[B \in [0,t_1]  \geq \frac{1}{2}(t_0-t_1).
 $$
As a result,
$$
f(A_{t_0}) +f(A_{t_1}) \geq \frac{1}{2}
[ OPT([t_1,t_0])- t_0+t_1 +t_0-t_1] =\frac{1}{2} OPT([t_1,t_0])].
$$

Hence, the best out of both of these sets provides at least $\frac{1}{4}OPT([t_1,t_0])$ principal utility. Note that $OPT$ gets at most $\Pr[B\in [t_1,t_0]]OPT([t_1,t_0])$ utility from this event therefore the best of these threshold provides a $4$-approximation to the events' contribution to $OPT$'s principal utility.

We have shown that there exists a threshold that approximates the utility of $OPT$ conditioned that $B\in [t_1,t_0]$ that is
$$\Pr[B\in [t_1,t_0]]OPT([t_1,t_0]).$$
 Let us focus on the remaining principal utility that is obtained by $OPT$ in the event that $B\in [0,t_2]$ where $t_2$ bias of the most-biased action smaller than $t_1$. One key observation is that this event happens with probability at most $1/2$ by our choice of $t_1$ since $\Pr[B\in [t_1,t_0]]\geq \frac{1}{2}$.
 
 If we consider the conditional distributions on this event we can repeat the same analysis to prove that there exists bias threshold $t_3<t_2$ such that
 $Prob[B\in [t_3,t_2)\mid B\in [0,t_2] ]\geq 1/2$ and also 
 $$\max\{f(A_{t_2}),f(A_{t_3})\} \geq \frac{1}{4} \Pr[B\in [0,t_2]] OPT([t_3,t_2]) \geq \frac{1}{4} \Pr[B\in [t_3,t_2]] OPT([t_3,t_2]), $$ which corresponds to $4$-approximation to the contribution to $OPT$ solution's principal utility in this interval. Note that   
 
$$\Pr[B\in [0,t_3)\mid B\in [0,t_2] ]\leq 1/2 \Rightarrow \Pr[B\in [0,t_3] ]\leq 1/4.$$
Repeating this process shrinks the probability of the remaining probability space by half. Let $m$ be the maximum number of times we can repeat this process. There are two ways this process can stop. Either we are left with a single action or the probability that $B$ ($OPT$ bias) is strictly below the last used threshold is $0$. Since the minimum probability of any realization is $p_{min}$ and each time we repeat this process the probability is shrunk by half $m$ cannot be larger than $\log{p_{min}^{-1}}$.

 This process generates $m$ disjoint events $B\in [t_{2i+1},t_{2i}]$ for $i=0,\dots,m-1$ such that

$$
OPT= \sum_{i=0}^{m-1} OPT([t_{2i+1},t_{2i}])Prob[B\in [t_{2i+1},t_{2i}]]
$$
and in addition 
$$
\max \{ f(A_{t_{2i+1}},A_{t_{2i}}\} \geq\frac{1}{4}
OPT([t_{2i+1},t_{2i}])\Pr[B\in [t_{2i+1},t_{2i}]]
$$
If we combined these two equations together we get that 
$$
\max_{i\in \{0,\dots,2m-1\}} f(A_{t_i}) \geq \frac{m}{4} OPT
$$
Since $m\leq \log{p_{min}^{-1}}$ we get that the best possible threshold is at least a $4\log{p_{min}^{-1}}$ approximation.


\end{proof}

\paragraph{Matching lower bound.}  
The above analysis is tight. We show in Appendix~\ref{supp:lognlb} that our analysis in Theorem~\ref{thm:uniform} is tight up to a constant factor. We do so by providing instances where no threshold policy can outperform the logarithmic approximation ratio.

\begin{theorem}\label{thm:lognlb}
	There exists a family of instances where no threshold policy is better than a $\Omega(\log \pmin^{-1})$-approximation.
\end{theorem}

In Appendix~\ref{app:hardness}, we also prove the following supplementary hardness result for the case with discretely-distributed, correlated values:
\begin{theorem}
	\label{thm:APX_hardness}
	With correlated, discrete values, there exists a constant $c$ such that it is NP-hard to compute a mechanism with approximation factor better than $c$.
\end{theorem}
To prove this result, we reduce from bounded-degree vertex cover, which is similarly hard to approximate.

\subsection*{Acknowledgments}
The authors thank Bobby Kleinberg, Hamsa Bastani, Rediet Abebe, and Shaddin Dughmi for helpful discussions. 
The authors further thank Yuanzhe Liu, Yichi Zhang, and Pucheng Xiong for their help as undergraduate research assistants on this and related projects.
Part of this work was done while Ali Khodabakhsh and Emmanouil Pountourakis were visiting the Simons Institute for the Theory of Computing.

\bibliographystyle{ACM-Reference-Format}
\bibliography{sample-bibliography}

\appendix
\section{Supplementary materials}

\subsection{Proof of Theorem~\ref{thm:indephard}}
\label{app:indephard}

To see that the problem is in NP, note that given an action set $S$, we may compute the principal's expected utility in polynomial time. Specifically, let $S_i$ be the first $i$ elements of $S$, let $\textsc{opt}_i(u)=\mathbb E[v_{g(S_i,\mathbf v)}~|~v_{g(S_i,\mathbf v)}+b_{g(S_i,\mathbf v)}=u]$ and $p_i(u)=\text{Pr}[v_{g(S_i,\mathbf v)}+b_{g(S_i,\mathbf v)}=u]$. Both can be computed for all $u$ by a simple dynamic program, considering $i=1,\ldots,n$. The utility of $S$ can then be computed using the law of total expectation.

To show hardness, we will reduce from \textsc{Integer Partition}.
An instance of this problem is integers $c_1,\ldots,c_n$. 
The goal is to find a subset $S\subseteq [n]$ such that $\sum_{i\in S}c_i=\tfrac{1}{2}\sum_{i=1}^n c_i$. 
Let $C=\sum_{i=1}^n c_i$, and $c_{\max}=\max_i c_i$.
Consider the following delegation instance, with $n+1$ actions. 
\begin{itemize}
	\item Actions $1,\ldots,n$ have bias $M^2(1-C/2M)$ for some $M$ (large, to be chosen). Let $\delta$ be a number small enough to only matter for agent tiebreaking (and which we will omit from all principal utility computations). The value of action $i$ will be
	\begin{itemize}
		\item[*] {\em (high realization)} $1+2\delta$ with probability $p_i=\tfrac{c_i}{M^3}-\tfrac{c_i^2}{2M^4(1-C/2M)}$
		\item[*] {\em (low realization)} $1$ with probability $q_i=\tfrac{c_i}{M}$
	\end{itemize}

	\item Action $n+1$ has value $0$ with probability $1/2$ and with probability $1/2$ takes value $M^2(1-C/2M)+1+\delta$. Action $n+1$ will have bias $0$.
\end{itemize}

First observe that any set of actions not containing action $n+1$ is suboptimal.
By a union bound, the utility from such a set is at most $\sum_{i=1}^n (p_i+q_i)\leq 2C/M$.
The utility from taking action $n+1$ alone, meanwhile, is $M^2(1-C/2M)+1\geq 1$, which is at least $2C/M$ as long as $M\geq 2C$.
It follows that the principal's problem is to pick which of actions $1,\ldots,n$ to pick alongside $n+1$.
Now consider the principal utility from a set $T=S\cup\{n+1\}$ for some $S\subseteq[n]$. To compute the principal's utility, consider the following two events:
\begin{itemize}
	\item Let $\mathcal E_1$ be the event that at least one action in $S$ has a high realization. In this event, the agent will choose such an action over $n+1$, no matter the value of action $n+1$. 
	
	We can approximate the probability of this event using only first-order terms. In more detail, this event has probability
	\begin{equation*}
		1-\prod_{i\in S}(1-p_i)=\sum_{i\in S}p_i-\sum_{k=2}^{|S|}(-1)^k\sum_{\substack{S_k\subseteq S:\\ |S_k|=k}}\prod_{j\in S_k}p_j
	\end{equation*}
Call the second term on the right $C_1$. We can show that $C_1\in[-\frac{4n^2c_{\max}^2}{M^6},\frac{4n^2c_{\max}^2}{M^6}]$:
\begin{align*}
	|C_1|&=
	\Bigg|\sum_{k=2}^{|S|}(-1)^k\sum_{\substack{S_k\subseteq S:\\ |S_k|=k}}\prod_{j\in S_k}p_j\Bigg|\\
	&\leq \sum_{k=2}^n\sum_{\substack{S_k\subseteq [n]:\\ |S_k|=k}}\prod_{j\in S_k}p_j\\
	&\leq \sum_{k=2}^n\binom{n}{k}\left(\frac{c_{\max}}{M^3}\right)^k\\
	&\leq \sum_{k=2}^n\left(\frac{ne}{k}\right)^k\left(\frac{c_{\max}}{M^3}\right)^k\\
	&\leq 2\sum_{k=2}^n\left(\frac{nc_{\max}}{M^3}\right)^k\\
	&\leq 2\sum_{k=2}^\infty\left(\frac{nc_{\max}}{M^3}\right)^k\\
	&=\frac{2n^2c_{\max}^2}{M^4}\sum_{k=0}^\infty\left(\frac{nc_{\max}}{M^3}\right)^k\\
	&=\frac{2n^2c_{\max}^2}{M^6}\frac{1}{1-\tfrac{nc_{\max}}{M^3}}.
\end{align*}
As long as $nc_{\max}/M^3\leq 1/2$, we have the desired upper bound.

\item Let $\mathcal E_2$ be the event that no action in $S$ has a high realization, and at least one has a low realization. Then this event has probability:
\begin{align*}
	&\prob[\overline{\mathcal E_1}]-\prod_{i\in S}(1-p_i-q_i)\\
	&=1-\sum_{i\in S}p_i+C_1-\prod_{i\in S}(1-p_i-q_i)\\
	&=1-\sum_{i\in S}p_i+C_1-1+\sum_{i\in S}(p_i+q_i)-\sum_{i\neq j\in S}(p_i+q_i)(p_j+q_j)+\sum_{k=3}^{|S|}(-1)^k\sum_{\substack{S_k\subseteq S:\\ |S_k|=k}}\prod_{j\in S_k}(p_j+q_j)\\
	&=C_1+\sum_{i\in S}q_i-\sum_{i\neq j\in S}(p_i+q_i)(p_j+q_j)+\sum_{k=3}^{|S|}(-1)^k\sum_{\substack{S_k\subseteq S:\\ |S_k|=k}}\prod_{j\in S_k}(p_j+q_j).
\end{align*}
Call the last term $C_2$. A similar argument to the one for $C_1$ shows that $C_2\in [-\frac{16n^3c_{\max}^3}{M^3},\frac{16n^3c_{\max}^3}{M^3}]$. We include it below for completeness.
\begin{align*}
	|C_2|&=\Bigg|\sum_{k=3}^{|S|}(-1)^k\sum_{\substack{S_k\subseteq S:\\ |S_k|=k}}\prod_{j\in S_k}(p_j+q_j)\Bigg|\\
	&\leq \sum_{k=3}^{|S|}\sum_{\substack{S_k\subseteq S:\\ |S_k|=k}}\prod_{j\in S_k}(p_j+q_j)\\
	&\leq \sum_{k=3}^{n}\sum_{\substack{S_k\subseteq S:\\ |S_k|=k}}\prod_{j\in S_k}(p_j+q_j)\\
	&\leq \sum_{k=3}^{n}\binom{n}{k}\left(\frac{2c_{\max}}{M}\right)^k\\
	&\leq \sum_{k=3}^{n}\left(\frac{ne}{k}\right)^k\left(\frac{2c_{\max}}{M}\right)^k\\
	&\leq \sum_{k=3}^{n}\left(\frac{2nc_{\max}}{M}\right)^k\\
	&\leq\frac{8n^3c_{\max}^3}{M^3}\frac{1}{1-\tfrac{2nc_{\max}}{M}}
\end{align*}
This implies the desired bound as long as $M\geq 4nc_{\max}$.

The third term above can also be simplified:
\begin{equation*}
	\sum_{i\neq j\in S}(p_i+q_i)(p_j+q_j)=\sum_{i\neq j\in S}p_ip_j+2\sum_{i\neq j\in S}p_iq_j+\sum_{i\neq j\in S}q_iq_j.
\end{equation*}
Call the first two terms above $C_3$. Since for any $i$, $0\leq p_i\leq c_{\max}/M^3$ and $0\leq q_i\leq c_{\max}$, we have $C_3\in[0,\frac{3n^2c_{\max}}{M^4}]$. We therefore have $\prob[\mathcal E_2]=\sum_{i\in S}q_i-\sum_{i\neq j\in S}q_iq_j+C_1+C_2-C_3$.
\end{itemize}

\paragraph{Analyzing the Principal's Utility}
Given actions $S$, we can write the principal's utility as:
\begin{equation*}
	\frac{1}{2}\left((1+2\delta)\prob[\mathcal E_1]+(M^2(1-C/2M)+1+\delta)\prob[\overline{\mathcal E_1}]\right)+\frac{1}{2}(\prob[\mathcal E_2]+(1+2\delta)\prob[\mathcal E_1]).
\end{equation*}
Taking $\delta\rightarrow 0$, we obtain:
\begin{align*}
	&\frac{1}{2}\left(\prob[\mathcal E_1]+(M^2(1-C/2M)+1)\prob[\overline{\mathcal E_1}]\right)+\frac{1}{2}(\prob[\mathcal E_2]+\prob[\mathcal E_1]).\\
	&\quad\quad\quad\quad=\frac{1}{2}(M^2(1-C/2M)+1)+\frac{1}{2}(\prob[\mathcal E_2]+\prob[\mathcal E_1]-M^2(1-C/2M)\prob[\mathcal E_1])\\
	&\quad\quad\quad\quad=\frac{1}{2}(M^2(1-C/2M)+1)+\frac{1}{2}(\prob[\mathcal E_2]-M^2(1-C/2M)\prob[\mathcal E_1])+\prob[\mathcal E_1]/2.
\end{align*}

The first term does not depend on $S$, and the last term will turn out to be negligibly small. We next analyze the middle term, leaving out $C_1$, $C_2$, and $C_3$ for the moment.
\begin{align*}
	\prob[\mathcal E_2]-M^2(1-C/2M)\prob[\mathcal E_1]&\approx\sum_{i\in S}\frac{c_i}{M}-\sum_{i\neq j\in S}\frac{c_ic_j}{M^2}-\left(1-\frac{C}{2M}\right)\left(\sum_{i\in S}\frac{c_i}{M}-\sum_{i\in S}\frac{c_i^2}{2M^2(1-\tfrac{C}{2M})}\right)\\
	&=\frac{C}{2M}\sum_{i\in S}\frac{c_i}{M}-\sum_{i\neq j\in S}\frac{c_ic_j}{M^2}-\sum_{i\in S}\frac{c_i^2}{2M^2}\\
	&=\frac{1}{2M^2}\left(\sum_{i\in S}c_i+\sum_{j\notin S}c_j\right)\sum_{i\in S}c_i-\sum_{i\neq j\in S}\frac{c_ic_j}{M^2}-\sum_{i\in S}\frac{c_i^2}{2M^2}\\
	&=\frac{1}{2M^2}\left(\sum_{i\in S}c_i\right)\left(\sum_{j\notin S}c_i\right).
\end{align*}

This latter expression takes value $C^2/8M^2$ if the integers can be exactly partitioned, and value at most $(C^2/4-1)/2M^2=C^2/8M^2-1/2M^2$ otherwise. Now we can take $C_1$, $C_2$, $C_3$, and $\prob[\mathcal E_1]/2$ into account.
Specifically, we can write:
\begin{align*}
	&\Bigg|\frac{1}{2}(\prob[\mathcal E_2]-M^2(1-C/2M)\prob[\mathcal E_1])+\prob[\mathcal E_1]/2-\frac{1}{2M^2}\left(\sum_{i\in S}c_i\right)\left(\sum_{j\notin S}c_i\right)\Bigg|\\
	&\quad\quad\quad\quad=\Big|\frac{1}{2}(C_1+C_2-C_3)+\frac{M^2}{2}(1-C/2M)C_1+\prob[\mathcal E_1]/2\Big|\\	
	&\quad\quad\quad\quad\leq\frac{1}{2}\left(\frac{4n^2c_{\max}^2}{M^6}+\frac{16n^3c_{\max}^3}{M^3}+\frac{3n^2c_{\max}}{M^4}\right)+\frac{M^2}{2}(1-C/2M)\frac{4n^2c_{\max}^2}{M^6}+\frac{nc_{\max}}{2M^3}\\
	&\quad\quad\quad\quad\leq \frac{16n^3c_{\max}^3}{M^3}.
\end{align*}
The first inequality follows from applying the triangle inequality, along with our existing bounds on $C_1$, $C_2$, and $C_3$ and the fact that $\prob[\mathcal E_1]\leq \frac{nc_{\max}}{M^3}$ by a union bound. As long as $M\geq 128n^3c_{\max}^3$, we will have that $16n^3c_{\max}^3/M^3\leq 1/8M^2$. We can therefore solve our \textsc{Integer Partition} instance by asking if our constructed instance of delegation has value at least $(M^2(1-C/2M)+1)/2+C^2/8M^2-1/8M^2$. Any solution that exactly partitions the integers will obtain at least this value, and any solution that fails to do so will have objective value at most $(M^2(1-C/2M)+1)/2+C^2/8M^2-1/2M^2+1/8M^2\leq(M^2(1-C/2M)+1)/2+ C^2/8M^2-1/8M^2$.
\subsection{Proof of Theorem~\ref{thm:lognlb}}
\label{supp:lognlb}
In this appendix, we show that the logarithmic approximation upper bond of Theorem~\ref{thm:uniform}  is tight, up to a constant factor. That is, no threshold algorithm can perform better than $\log p_{min}^{-1}$.
To prove the tightness of our analysis in Section~\ref{sec:res_distortion}, we construct an infinite family of instances.


For any $k\geq 2$, consider an instance with $n=2k-1$ actions. The correlated distribution has $m=2^k-1$ value profile realizations. We construct a value matrix where each row correspond to an action and each column corresponds to a realization. Therefore, the value at cell $V_{i,j}$ gives the value of action $i$ at realization value profile $j$. The distribution over value profiles simply selects and value realization uniformly at random.

\footnotesize
\begin{align}
&V=\label{eq:log_up}\\
&\left[ 
\begin{array}{c|c|c|c|c|c|c|c|c|c|c}
2^k & & & & &&&&&&\\
\hline 
2^{k-1}+2\epsilon &  &  &  &  &&&&&&\\
\hline 
 & 
 2^{k-1} & 
 2^{k-1} &  &  &&&&&&\\
\hline 
2^{k-2}+3\epsilon & 2^{k-2}+3\epsilon & 2^{k-2}+3\epsilon & &&&&&&&\\
\hline 
 &  &  & 
 2^{k-2} & 
 2^{k-2} & 
 2^{k-2} & 
 2^{k-2} &&&&\\
\hline 
2^{k-3}+4\epsilon & 2^{k-3}+4\epsilon & 2^{k-3}+4\epsilon & 2^{k-3}+4\epsilon & 2^{k-3}+4\epsilon & 2^{k-3}+4\epsilon & 2^{k-3}+4\epsilon &&&&\\
\hline 
\vdots& \vdots & \vdots & \vdots &\vdots &\vdots&\vdots&\ddots&&&\\
\hline
2+O(k\epsilon)& 2+O(k\epsilon) & 2+O(k\epsilon) & 2+O(k\epsilon) & 2+O(k\epsilon)&2+O(k\epsilon)&2+O(k\epsilon)&\dots&&&\\
\hline
 && && &&&\dots&\undermat{2^{k-1}}{
2&
\dots&
2}
\end{array}\right]\notag\\
 & \notag
\end{align}
\normalsize

Note that all the empty entries in the above matrix are zero, and are removed to make the structure of the matrix more apparent. Also, 
 every solution corresponds to eliminating a  a set of rows. For each realization (column) the row with maximum agent utility is selected. 
 The optimal solution is to select set of odd actions (with size $k$). The colored entries indicate (realization,action) pairs that contribute to the optimal principal's utility ($OPT$). The optimal utility is equally divided between the odd rows, leaving $2^k$ for each one: the first row has $2^k$ in the first column, the third row has $2^{k-1}$ in columns 2 and 3, the fifth row has $2^{k-2}$ over the next $4$ columns and so on. In the example, the even actions are constructed to lower the principal's utility whenever they are included in a threshold solution.  In every state (column), we divide the colored utility by 2 to find the utility of the next row, and keep dividing by 2 to complete the subsequent even rows. The $\epsilon$ terms are added to break the ties and are of little importance.

Next, we define the bias: we set $b_1=0$, and the rest of actions have the following bias:
\begin{equation}
\label{eq:log_distortion}
b_{2i+1}=\sum_{j=1}^i 2^{k-j},\quad b_{2i}=b_{2i+1}-\epsilon, \qquad i\in \{1,...,k-1\}.
\end{equation}
 Now that all the parameters are set, it is easy to verify that given the set of odd actions (rows), the agent will indeed pick the colored entries. This generates the optimal utility, since it is optimal in every single realization. Since each value profiled is realized with probability $\frac{1}{m}$ the optimal expected utility is equal to: 
$$OPT=\frac{k\times 2^k}{m}.$$
However, the best threshold solution in the constructed instance is to allow the entire set of actions ($\Omega$).
To see this, assume that the principal allows actions with bias less than or equal to $b_{2\ell-1}$ for some $\ell\leq k$. (Thresholds set at even-indexed actions can be easily shown to be suboptimal.)
Note that every even action is preferred by the agent to any other action with less bias. Therefore, the only actions chosen by the agent are $2\ell -1$ or $2\ell -2$ 
In this case, the principal will get utility of $2^{k-\ell+1}+O(\ell\epsilon)$ from the first $2^\ell-1$ states, and zero from the remaining states. 

Observe that the overall utility $(2^\ell-1)\times 2^{k-\ell+1}$ is an increasing function in $\ell$, meaning that the best strategy for the principal is to not limit the agent.
In this case, the agent will pick the penultimate action in the first half of columns, and the last action for the second half, generating utility of (almost) 2 for principal in every state. More precisely, we have:
$$APX=2+O(k\epsilon).$$
We get the desired lower bound by dividing the above objectives:
$$\frac{OPT}{APX}=\frac{k\times 2^k}{2n+O(nk\epsilon)}\cong \frac{k}{2}\cong \frac{\log{n}}{2}.$$
\begin{example} In order to make sure that the above construction is clear, here we present the full matrices for the case of $k=3$, which translates into $n=5$ actions and $m=7$ realizations. The value matrix in this case is
\begin{equation*}
V=\left[ 
\begin{array}{c|c|c|c|c|c|c}
8 &0 &0 &0 &0 &0 &0\\
\hline 
4+2\epsilon &0  &0  &0  &0  &0 &0\\
\hline 
0 & 
4 & 
4 &0  &0  &0 &0\\
\hline 
2+3\epsilon & 2+3\epsilon & 2+3\epsilon &0 &0 &0 &0\\
\hline 
0 &0  &0  & 
2 & 
2 & 
2 & 
2 \end{array}\right ]
\end{equation*}
Calculating the bias in  \eqref{eq:log_distortion} results in:

$$\mathbf{b}=(0,4-\epsilon,4,6-\epsilon,6)$$
It is clear that the value matrix $V$ is non-negative, and the agent's utility $V+B$ will be:
\begin{equation*}
V+B=\left[ 
\begin{array}{c|c|c|c|c|c|c}
8 &0 &0 &0 &0 &0 &0\\
\hline 
8+\epsilon &4-\epsilon  &4-\epsilon  &4-\epsilon  &4-\epsilon  &4-\epsilon &4-\epsilon\\
\hline 
4 & 8 & 8 &4  &4  &4 &4\\
\hline 
8+2\epsilon & 8+2\epsilon & 8+2\epsilon &6-\epsilon &6-\epsilon &6-\epsilon &6-\epsilon\\
\hline 
6 &6  &6  & 8 & 8 & 8 & 8 \end{array}\right ]
\end{equation*}
Observe that $OPT=24/7$ by the set of odd actions $\{1,3,5\}$, while $APX=(14+9\epsilon)/7$ from the entire set of actions $\Omega=\{1,2,3,4,5\}$.
\end{example}

\subsection{Proof of Theorem~\ref{thm:APX_hardness}}
\label{app:hardness}

\begin{proof}

We give a reduction from the bounded degree vertex cover problem, i.e., the vertex cover problem on graphs with degree at most $B$ (constant). This problem is known to be APX-hard \cite{clementi1999improved}. Consider an instance of the bounded degree vertex cover problem $\Gt=(\Vt,\Et)$ with $\nt$ nodes and $\mt$ edges (where $\mt\leq B\cdot\nt/2=\O(\nt)$).\footnote{To distinguish between the parameters of the vertex cover instance and the delegation instance, we use tilde ($\sim$) for the graph instance.}  

We construct an instance of the delegation problem with
$\nt+1$ actions with action $a_i$ corresponding to node $i$ and an additional ``default'' action $a_0$. All actions have $0$ bias apart from $a_0$ which has bias $-1$. The correlated distribution of the actions values is defined as follows: we pick an edge $e=\{i,j\}$ or some node $i$ uniformly at random, i.e., each element with probability $(\mt +\nt)^{-1}$

If we picked some edge $e=\{i,j\}$  then  we assign value $5$ to actions $a_i$ and $a_j$, $2$ to the default action $a_0$, and $0$ for all  other actions. If we picked a node $i$ we assign value $2$ to $a_i$ and $a_0$ (default action) and $0$ for all other actions.

We claim that the optimal solution of the delegation problem produces a utility of $(5\mt+3\nt-\kt)/(\mt+\nt)$ for the principal, where $\kt$ is the size of the smallest vertex cover of $\Gt$. To see this, first note that any solution $\St \subseteq \Vt$ can be improved by including $a_0$, since $a_0$ has {a negative bias}. Any time the agent would choose $a_0$, it is the optimal choice for the principal as well. We therefore only consider solutions containing $a_0$.

Now if $\St$ is a vertex cover of $\Gt$ with $\abs{\St}=\kt$, consider the corresponding delegation set where the principal allows actions $\{a_i: i\in \St\}\cup\{a_0\}$. If we generate the values by picking an edge, the agent will pick the action corresponding to one end of that edge (one is guaranteed to be in the cover $\St$) to get a utility of $5$ compared to $2-1$ achievable from the default action. This choice will also generate utility of $5$ for the principal, which makes $5\mt$ in total. If the utility is generated by picking node $i$ the agent will pick action $a_i$ which generates the utility of $2$ for both principal and agent. This will make $2\kt$ in total. Finally, if the utilities are generated using some node $i \in \Vt\backslash\St$ the agent picks the default action which generates a utility of $2-1$ for the agent but $2+1$ for the principal. This will give $3(\nt-\kt)$ in total. As a result the principal utility in expectation is $(5\mt+3\nt-\kt)/(\mt+\nt)$.

For the converse, consider an optimal solution $A$ to the delegation problem. We show that the nodes corresponding to the actions in $A$ (excluding the default action) induce a vertex cover; otherwise the solution can be improved. Assume that there exists an edge $e=\{i,j\}$ where neither $a_i$ nor $a_j$ is allowed in $A$. If we add action $a_i$ to $A$, the principal gets a utility of $5$ if the utilities are generated from pick edge $e$, compared to current utility of $3$ from the default action. On the other hand, the utility of the principal decreases from $3$ to $2$ if the values are generated by action $i$. So the total utility of $A\cup\{a_i\}$ is more than $A$ which contradicts the optimality of $A$. Therefore $A$ should be a vertex cover (plus default action). This in turn implies that the utility is at most $(5\mt+3\nt-\kt)/(\mt+\nt)$ where $\kt$ is the size of the minimum vertex cover.

Since $\mt=\Theta(\nt)$ and the minimum vertex cover has size at least $\mt/B=\Omega(\nt)$, a constant factor gap in the bounded degree vertex cover problem translates into a constant factor gap in the optimal solution of the delegation problem, which yields the desired hardness result. 
\end{proof}

\end{document}